\documentclass[12pt]{article}
\usepackage[margin=1in]{geometry}
\usepackage{amsthm} 
\usepackage{setspace} 
\usepackage{hyperref}
\newtheorem{theorem}{Theorem}
\newtheorem{lemma}{Lemma}
\newtheorem{corollary}{Corollary}

\usepackage{amssymb, amsmath, amsfonts, natbib, graphicx, color, bm, float, multicol, multirow}
\usepackage{bbm} 
\usepackage{enumitem} 
\newfloat{algorithm}{htbp}{loa}
\floatname{algorithm}{Algorithm}

\DeclareMathOperator{\E}{E}

\newcommand{\algline}{\rule{\textwidth}{0.2mm}}

\title{Lazy ABC}

\author{Dennis Prangle\footnote{University of Reading.  Email \href{mailto:d.b.prangle@reading.ac.uk}{\nolinkurl{d.b.prangle@reading.ac.uk}}}}
\date{}

\begin{document}

\maketitle

\begin{abstract}
Approximate Bayesian computation (ABC) performs statistical inference for otherwise intractable probability models by accepting parameter proposals when corresponding simulated datasets are sufficiently close to the observations.  Producing the large quantity of simulations needed requires considerable computing time.  However, it is often clear before a simulation ends that it is unpromising: it is likely to produce a poor match or require excessive time.  This paper proposes lazy ABC, an ABC importance sampling algorithm which saves time by sometimes abandoning such simulations.  This makes ABC more scalable to applications where simulation is expensive.  By using a random stopping rule and appropriate reweighting step, the target distribution is unchanged from that of standard ABC.  Theory and practical methods to tune lazy ABC are presented and illustrated on a simple epidemic model example. They are also demonstrated on the computationally demanding spatial extremes application of \cite{Erhardt:2012}, producing efficiency gains, in terms of effective sample size per unit CPU time, of roughly 3 times for a 20 location dataset, and 8 times for 35 locations.
\end{abstract}

\section{Introduction}

Approximate Bayesian computation (ABC) algorithms are a popular method of inference for a wide class of otherwise intractable probability models in applications such as population genetics, ecology, epidemiology and systems biology \citep{Beaumont:2010, Marin:2012}.  They select parameter vectors $\theta$ for which datasets $y$ simulated from the model of interest are sufficiently close to the observations.  A bottleneck is the computational cost of producing the large quantity of model simulations needed, which becomes increasingly severe for more detailed models.  However, it is often clear during a simulation that it is unpromising.  For example it is likely to produce a poor match or to require excessive computation time.  This paper presents \emph{lazy ABC}, an importance sampling method which abandons some such simulations, a step referred to as \emph{early stopping}, exploiting information from incomplete simulations to save time.  The result is an ABC algorithm which is more scalable to applications where simulation is computationally demanding. 

In more detail, standard ABC is based on a random likelihood estimator $\hat{L}_{\text{ABC}}$, which is 1 for a close match of simulated and observed data, and zero otherwise.  The algorithm can be shown to target a distribution corresponding to the approximate likelihood $L_{\text{ABC}}(\theta) = \E[\hat{L}_{\text{ABC}} | \theta]$.  Lazy ABC is based on an alternative estimator $\hat{L}_{\text{lazy}}$.  This equals zero with probability $1-\alpha$ -- if early stopping is performed -- and otherwise equals the $\hat{L}_{\text{ABC}}$ estimator multiplied by a weight.  Letting the weight equal $1/\alpha$ makes $\hat{L}_{\text{lazy}}$ an unbiased estimator of $L_{\text{ABC}}(\theta)$.  Results on random likelihood estimates show that importance sampling 
algorithms based on $\hat{L}_{\text{lazy}}(\theta)$ therefore target the same distribution as standard ABC.  No further approximation has been introduced.

The lazy ABC estimator trades off an increase in variance for a reduction in computation time.  It is shown that for this to be most advantageous $\alpha$ should be (1) larger when there is a high probability of the simulated dataset being a close match to the observations (2) smaller when the expected time to complete the simulation is large.  To achieve this, $\alpha$ is based on $X$, a random variable encapsulating some preliminary information about the simulated dataset $Y$, also a random variable.  The final likelihood estimator is based on $X$ and $Y$.  However when early stopping occurs a realised value of zero is obtained without drawing a value of $Y$. This is an example of the computer science concept of \emph{lazy evaluation} \citep[``functional programming'' entry]{Ralston:2003}, which is the basis for the method's name.

The paper presents theoretical results on the optimal tuning of $\alpha$ in lazy ABC, making precise the two properties just outlined.  This choice is asymptotically optimal in terms of maximising efficiency, which is defined as effective sample size (ESS) per unit CPU time.  Based on this, a framework for tuning in practice is also presented.  The main requirements are the estimation of the probability of close matches and of expected remaining computation times.  Both of these are conditional on $\theta$ and $x$ (realised values of $X$).  As $x$ is typically high dimensional this estimation is not feasible, so instead it is recommended to base the choice of $\alpha$ on a vector of low dimensional \emph{decision statistics} $\phi(\theta, x)$.  A computationally demanding example is presented based on the spatial extremes application of \cite{Erhardt:2012} where lazy ABC increases the efficiency by roughly 3 times for modestly sized data and 8 times for a larger example.

It should be emphasised that this approach to tuning is \emph{optional}.
In many applications a simple ad-hoc choice of $\alpha$ may produce significant efficiency gains,
especially where application knowledge can be used to create heuristic criteria which quickly identify many unpromising simulations.

The focus of this paper is on importance sampling, which is widely used by ABC practitioners and particularly amenable to parallelisation.  However the lazy ABC approach is also applicable to other algorithms, such as Markov chain Monte Carlo (MCMC) and sequential Monte Carlo (SMC), as discussed in the final section.

Several recent papers have proposed speeding up ABC by fitting a model to $(\theta, y)$ pairs, simulated either in a preliminary stage or in earlier ABC iterations \citep{Buzbas:2013, Meeds:2014, Moores:2014, Wilkinson:2014}.  This model is then sometimes or always used in place of the original model of interest in the inference algorithm.  A potential application of lazy ABC is to make use of such approximate models (their predictions given $\theta$ forming $X$ in the notation above) to gain speed benefits without incurring additional approximation errors.  More generally, there has been much interest over the past decade in Bayesian inference algorithms with random weights \citep[e.g.][]{Beaumont:2003, Andrieu:2009, Fearnhead:2010, Tran:2014}.  A novelty of lazy ABC is that it introduces a random factor to the weights to reduce computation time, rather than to deal with intractability.  


Finally, a byproduct of the paper's theory, Corollary \ref{cor:ABC-IS}, is of independent interest.
This gives importance densities which optimise asymptotic ESS per unit CPU time for:
(a) ABC importance sampling (b) importance sampling with random weights.

The remainder of the paper is structured as follows.  Section \ref{sec:background} contains background material on ABC and importance sampling.  
Section \ref{sec:ESalgorithm} gives the lazy ABC algorithm and proves it targets the correct distribution.  Section \ref{sec:tuning} presents theory and practical methods for tuning the algorithm, as well as the corollary mentioned above.
Section \ref{sec:epidemics} illustrates the method for a simple epidemic model.
R code to implement this is available at
\texttt{https://github.com/dennisprangle/lazyABCexample}.
Section \ref{sec:extremes} contains a more challenging spatial extremes application and
Section \ref{sec:discussion} is a discussion.
Appendices contain proofs and discuss some extensions:
lazy ABC with multiple stopping decisions
and application of similar methods outside ABC.

\section{Importance sampling} \label{sec:background}

Consider analysing data $y_{\text{obs}}$ under a probability model with density $\pi(y|\theta)$ and parameters $\theta$.  The likelihood is defined as $L(\theta) = \pi(y_{\text{obs}}|\theta)$.  Bayesian inference introduces a prior distribution with density $\pi(\theta)$ and aims to find the posterior distribution $\pi(\theta | y_{\text{obs}}) = \pi(\theta) L(\theta) / \pi(y_{\text{obs}})$, where $\pi(y_{\text{obs}}) = \int \pi(\theta) L(\theta) d\theta$, or at least to estimate the posterior expectation $\E[h(\theta) | y_{\text{obs}}]$ of a generic function $h(\theta)$.  Importance sampling is a simple method to do this.  Parameter values $\theta_{1:N}$ are simulated independently from an importance density $g(\theta)$ and given weights $w_i = L(\theta_i) \pi(\theta_i) / g(\theta_i)$ (n.b.~$\theta_{1:N}$ represents the sequence $(\theta_i)_{1 \leq i \leq N}$.  Similar notation is used later.)  It is assumed throughout that $g(\theta)>0$ whenever $\pi(\theta)>0$.  Each of the $(\theta_i, w_i)$ pairs can be computed in parallel, allowing for efficient implementation.

A Monte Carlo estimate of $\E[h(\theta) | y_{\text{obs}}]$ is $\mu_h = \frac{\sum_{i=1}^N h(\theta_i) w_i}{\sum_{i=1}^N w_i}$.
Two properties of importance sampling estimates are
\begin{align}
&\mu_h \to \E[h(\theta) | y_{\text{obs}}] \quad \text{almost surely as } N \to \infty, \label{eq:IS1} \\
&E[N^{-1} \sum_{i=1}^N w_i] = \pi(y_{\text{obs}}). \label{eq:IS2}
\end{align}
See \cite{Geweke:1989} for proof that \eqref{eq:IS1} holds under weak conditions.  To prove \eqref{eq:IS2} note that each $w_i$ is an unbiased estimator of $\pi(y_{\text{obs}})$.  Estimating this is of interest for model comparison.  

\subsection{Notation}

The remainder of the paper is largely concerned with the distribution of random variables produced in an iteration of various importance sampling algorithms.  Henceforth, expectations and probabilities that involve quantities produced by importance sampling should be read as being with respect to this distribution.  In particular this means that below the marginal density of $\theta$ is taken to be $g(\theta)$.  The preceding material in this section is the only time that a marginal density of $\pi(\theta)$ is used instead.

\subsection{Random weights} \label{sec:randomweightIS}

Algorithm \ref{alg:RW-IS} describes random weight importance sampling (RW-IS), an importance sampling algorithm in which likelihood evaluations are replaced with random estimates of the likelihood.  Under the condition that these estimates are non-negative and unbiased, the algorithm produces valid output, in the sense that \eqref{eq:IS1} and \eqref{eq:IS2} continue to hold.  This can be seen by noting that Algorithm \ref{alg:RW-IS} is equivalent to a deterministic weight importance sampling algorithm with augmented parameters $(\theta, \ell)$, prior density $\pi(\theta) \pi(\ell| \theta)$, importance density $g(\theta) \pi(\ell| \theta)$ and likelihood $\ell$.  Here $\ell$ is the realisation of the likelihood estimator and $\pi(\ell| \theta)$ is the conditional density of this estimator.  This algorithm gives the correct marginal posterior for $\theta$.
See  \cite{Tran:2014} for a more detailed proof and \cite{Fearnhead:2010} for further background
(including the observation that the non-negativity condition above can be removed.)

\begin{algorithm}[htp]

\algline \\
{\bf Input:}
\begin{itemize}
\item Prior density $\pi(\theta)$ and importance density $g(\theta)$.
\item Number of iterations to perform $N$.
\item Likelihood estimator $\hat{L}$.
\end{itemize}
Repeat the following steps $N$ times.
\begin{enumerate}
\item[1] Simulate $\theta^*$ from $g(\theta)$.
\item[2] Simulate $\ell^*$ from $\hat{L} | \theta$.
\item[3] Set $w^* = \ell^* \pi(\theta^*) / g(\theta^*)$.
\end{enumerate}
{\bf Output:}\\ A set of $N$ pairs of $(\theta^*, w^*)$ values.\\
\algline
\caption{Random weight importance sampling (RW-IS) \label{alg:RW-IS}}
\end{algorithm}

\begin{algorithm}[htp]

\algline \\
{\bf Input:}
\begin{itemize}
\item Prior density $\pi(\theta)$ and importance density $g(\theta)$.
\item Number of iterations to perform $N$.
\item Observed data $y_{\text{obs}}$.
\item Summary statistics $s(\cdot)$, distance function $d(\cdot,\cdot)$ and threshold $\epsilon \geq 0$.
\end{itemize}
{\bf Algorithm:}\\
Repeat the following steps $N$ times.
\begin{enumerate}
\item[1] Simulate $\theta^*$ from $g(\theta)$.
\item[2] Simulate $y^*$ from $Y | \theta^*$
\item[3] Set $\ell^* = \mathbbm{1}[d(s(y^*), s(y_{\text{obs}})) \leq \epsilon]$.
\item[4] Set $w^* = \ell^* \pi(\theta^*) / g(\theta^*)$.
\end{enumerate}
{\bf Output:}\\ A set of $N$ pairs of $(\theta^*, w^*)$ values.\\
\algline
\caption{ABC importance sampling (ABC-IS) \label{alg:ABC-IS}}
\end{algorithm}

\subsection{Approximate Bayesian computation}

Many interesting models are sufficiently complicated that it is not feasible to calculate exact likelihoods or useful (i.e.~reasonably low variance) unbiased estimators.  ABC algorithms instead base inference on simulation from the model.
Algorithm \ref{alg:ABC-IS} (ABC-IS) is a standard importance sampling implementation of this idea.

For later theoretical work, it is helpful to interpret ABC-IS in the framework of RW-IS.
In particular, Algorithm \ref{alg:ABC-IS} can be obtained from Algorithm \ref{alg:RW-IS} by replacing the likelihood estimator $\hat{L}$ with a Bernoulli estimator 
\begin{equation} \label{eq:LhatABC}
\hat{L}_{\text{ABC}} = \mathbbm{1}[d(s(Y), s(y_{\text{obs}})) \leq \epsilon].
\end{equation}
This equals 1 if the distance between summary statistics of the simulated and observed datasets is less than or equal to a threshold $\epsilon$.
It is typically a biased estimate of the likelihood and ABC-IS targets an approximate likelihood $L_{\text{ABC}}(\theta) = E[\hat{L}_{\text{ABC}} | \theta]$.

A special case of ABC-IS is when $g(\theta)=\pi(\theta)$.  The weights in this case equal zero or one, and it is often referred to as ABC rejection sampling.  A generalisation of ABC-IS, considered in Section \ref{sec:discussion}, is to use as a likelihood estimator $K(d(s(Y), s(y_{\text{obs}}))/\epsilon)$, where $K$ is a density function known as the ABC kernel.  Algorithm \ref{alg:ABC-IS} uses a uniform kernel.

As $\epsilon \to 0$, the target distribution of ABC-IS converges to $\pi(\theta | s(y_{\text{obs}}))$.  However, $\epsilon>0$ is typically required to achieve a reasonable number of non-zero weights, so a trade-off must be made.  The observed summary statistics $s(y_{\text{obs}})$ should ideally preserve most of the information on $\theta$ available from $y_{\text{obs}}$.  However analysis of ABC algorithms shows that the quality of the approximation deteriorates with the dimension of $s(y)$.  Therefore the choice of $s(\cdot)$ involves a trade-off between low dimension and informativeness.  For further background details on all aspects of ABC see the review articles of \cite{Beaumont:2010} and \cite{Marin:2012}.

\section{Lazy ABC} \label{sec:ESalgorithm}

Lazy ABC is Algorithm \ref{alg:lazyABC}. Given proposed parameters $\theta^*$, step 2 performs an initial part of data simulation, with results $x^*$. A continuation probability $\alpha(\theta^*, x^*)$ is calculated. With this probability the simulation is completed, resulting in $y^*$, and a weight is calculated by steps 5 and 6. Otherwise the iteration is given weight zero. The desired behaviour is that simulating $x^*$ is computationally cheap but can be used to quickly reject many unpromising importance sampling iterations.

\begin{algorithm}[htp]

\algline \\
{\bf Input:}
\begin{itemize}
\item Prior density $\pi(\theta)$ and importance density $g(\theta)$.
\item Number of iterations to perform $N$.
\item Observed data $y_{\text{obs}}$.
\item Summary statistics $s(\cdot)$, distance function $d(\cdot,\cdot)$ and threshold $\epsilon \geq 0$.
\item Continuation probability function $\alpha(\theta,x)$ taking values in $[0,1]$.
\item Random variable $X$ representing an initial stage of simulation.
\end{itemize}
{\bf Algorithm:}\\
Repeat the following steps $N$ times.
\begin{enumerate}
\item[1] Simulate $\theta^*$ from $g(\theta)$.
\item[2] Simulate $x^*$ from $X | \theta^*$ and let $a^* = \alpha(\theta^*, x^*)$.
\item[3] With probability $a^*$ continue to step 4. Otherwise perform \emph{early stopping}: let $\ell^*=0$ and go to step 6.
\item[4] Simulate $y^*$ from $Y | \theta^*, x^*$.
\item[5] Set $\ell^*_{\text{ABC}} = \mathbbm{1}[d(s(y^*), s(y_{\text{obs}})) \leq \epsilon]$ and $\ell^* = \ell^*_{\text{ABC}} / a^*$.
\item[6] Set $w^* = \ell^* \pi(\theta^*) / g(\theta^*)$.
\end{enumerate}
{\bf Output:}\\ A set of $N$ pairs of $(\theta^*, w^*)$ values.\\
\algline
\caption{Lazy ABC \label{alg:lazyABC}}
\end{algorithm}

The following conditions are required for Algorithm \ref{alg:lazyABC} to be well defined:
\begin{itemize}
\item[C1] $\alpha(\theta,x) > 0$ whenever $\Pr(\hat{L}_{\text{ABC}}>0 | \theta, x) > 0$
\item[] (recall $\hat{L}_{\text{ABC}}$ is defined by \eqref{eq:LhatABC})
\item[C2] The random variable $X$ is such that both $X | \theta$ and $Y | \theta, x$ can be simulated from.
\end{itemize}

Lazy ABC and ABC-IS both perform importance sampling inference for the same approximate likelihood $L_{\text{ABC}}(\theta)$.
To see this first observe that Algorithm \ref{alg:lazyABC} can be interpreted as a RW-IS algorithm using a likelihood estimator of the form
\begin{equation} \label{eq:es}
\hat{L}_{\text{lazy}} = \begin{cases} 
  \hat{L}_{\text{ABC}} / \alpha(\theta, X) & \text{with probability } \alpha(\theta, X) \\
  0 & \text{otherwise}
\end{cases}
\end{equation}
By the arguments of Section \ref{sec:background}, to show that lazy ABC and ABC-IS target the same likelihood it is sufficient to prove the following result.

\begin{theorem} \label{thm:es}
Conditional on $\theta$, $\hat{L}_{\text{lazy}}$ is a non-negative unbiased estimator of $L_{\text{ABC}}(\theta)$.
\end{theorem}

\begin{proof}
Non-negativity is immediate.  For unbiasedness first observe that $\E(\hat{L}_{\text{lazy}} | \theta, x)$ equals zero when $\alpha(\theta,x)=0$ and $E(\hat{L}_{\text{ABC}} | \theta, x)$ otherwise.
By C1 if $\alpha(\theta,x)=0$ then $\Pr(\hat{L}_{\text{ABC}}>0 | \theta,x)=0$ and so $\E(\hat{L}_{\text{ABC}} | \theta, x) = 0$.  Hence $\E(\hat{L}_{\text{lazy}} | \theta, X) = \E(\hat{L}_{\text{ABC}} | \theta, X)$.  Taking expectations over $X$ gives the required result.
\end{proof}

\subsection{Examples}

The following are examples of situations in which lazy ABC can be useful. The first three form the main focus of the paper.  Each assumes that $Y$ is a deterministic function of a latent vector $X_{1:p}$ such that it is possible to simulate from $X_1 | \theta$ and $X_i | \theta, x_{1:i-1}$ for all $1<i\leq p$.  Further examples are given in Appendix \ref{sec:multistop} which use the same framework to consider multiple stopping decisions.
\begin{itemize}[leftmargin=*]
\item[] \emph{Example 1: Partial simulation}  Let $X=X_{1:t}$ for some $t<p$.
\item[] \emph{Example 2: Partial calculation of $s$} Assume that computing $s(Y)$ involves calculating variables $X'_{1:q}$ which are deterministic transformations of $Y$, and that this is the most expensive part of simulating $\hat{L}_{\text{ABC}}$.  Let $X=(X_{1:p}, X'_{1:t})$ for some $t<q$.  (This is applied in Section \ref{sec:simmethods}.)
\item[] \emph{Example 3: Random stopping times} As for either previous example but with $t$ replaced by a random stopping time variable $T$.  This allows a stopping decision once a particular event has occurred.
\item[] \emph{Example 4: Deterministic stopping}
Suppose $s(Y)=(s_1(X_{1:t}), s_2(X_{t+1:p}))$.
Let $X=X_{1:t}$ and $\alpha(\theta,x) = \mathbbm{1}[\Pr(\hat{L}_{\text{ABC}}>0 | \theta, x) > 0]$. That is, early stopping occurs if and only if $s_1$ is too extreme for the ABC acceptance criterion to be met.
\end{itemize}

\subsection{Terminology and notation} \label{sec:termnot}

Simulation from $X | \theta$ is referred to as the \emph{initial simulation stage} and from $s(Y) | \theta, x$ as the \emph{continuation simulation stage}.  It is often useful later to have $\alpha(\theta, x) = \alpha(\phi(\theta, x))$ where $\phi(\theta, x)$ is a vector of summaries referred to as the \emph{decision statistics}.

Notation is now introduced for expected CPU times: $\bar{T}_1(\theta)$ is for steps 2 and 3 in Algorithm \ref{alg:lazyABC} conditional on $\theta$, $\bar{T}_2(\theta, \phi)$ is for steps 4--6 conditional on $(\theta, \phi)$ and $\bar{T}(\theta)$ is steps 2--4 of Algorithm \ref{alg:ABC-IS} conditional on $\theta$.  The first two are roughly the times of the initial simulation and continuation stages, but also cover the other steps of Algorithm \ref{alg:lazyABC} given $\theta$.
The following definitions are also used later: $\bar{T}_1 = \E[\bar{T}_1(\theta)]$ and $\bar{T}_2(\phi) = \E[\bar{T}_2(\theta, \phi) | \phi]$.
Note that $\bar{T}_1$ is a constant.
The other terms above are deterministic functions of one or both of $\theta$ and $\phi$ and
may therefore be random variables when their arguments are.



\section{Tuning} \label{sec:tuning}

There is considerable freedom to tune lazy ABC through the choice of $X$ (when to consider stopping) and $\alpha$ (the function assigning continuation probabilities).  Section \ref{sec:tuning theory} provides theory on the most efficient choice of $\alpha$.  This is used
in Section \ref{sec:tuning practice} to motivate practical tuning methods.


\subsection{Theory} \label{sec:tuning theory}

A commonly used tool for the analysis of importance sampling algorithms is the effective sample size (ESS).  \cite{Liu:1996} argued that typically the variance of the importance sampling estimator is roughly equal to that of $N_{\text{eff}}$ independent samples where
\[
N_{\text{eff}} = N E(W)^2 / E(W^2),
\]
$N$ is the number of importance sampling iterations
and the random variable $W$ is the weight generated in an iteration of importance sampling (with zero representing rejection).
The argument of \citeauthor{Liu:1996} generalises immediately to RW-IS algorithms through the interpretation of them as importance sampling algorithms on an augmented parameter space given in Section \ref{sec:randomweightIS}.

Define efficiency as $N_{\text{eff}} / T$ where $T$ is the CPU time of the algorithm (i.e.~ignoring any execution time savings due to parallelisation.)  
Various results on asymptotic efficiency for large $N$ now follow.
These are proved in Appendix \ref{sec:alpha} as corollaries of Theorem \ref{thm:tuning} which is stated there.

The choice of $\alpha(\theta,x)$ which maximises asymptotic efficiency is as follows for some $\lambda \geq 0$,
\begin{equation} \label{eq:tuning1}
\alpha(\theta,x) = \min\left\{1, \lambda \frac{\pi(\theta)}{g(\theta)} \left[\frac{\gamma(\theta,x)}{\bar{T}_2(\theta,x)}\right]^{1/2} 
\right\},
\end{equation}
where $\gamma(\theta, x) = \Pr \left(d(s(Y), s(y_{\text{obs}})) \leq \epsilon | \theta, x \right)$ and $\bar{T}_2(\theta, x)$ is as defined in Section \ref{sec:termnot}. Roughly, $\gamma(\theta, x)$ is the probability that continuing the simulation will meet the ABC acceptance criterion and $\bar{T}_2(\theta, x)$ is the expected time it will take. The interpretation of $\lambda$ is not clear, nor does a simple closed form expression for its optimal value seem possible.

This optimal choice of $\alpha$ is of limited practical use, as typically $X$ is high dimensional and thus estimation of $\gamma(\theta, x)$ and $\bar{T}_2(\theta, x)$ is not feasible. Estimation is more feasible if $(\theta, x)$ is replaced by some low dimensional vector of summaries, $\phi(\theta,x)$, which will be referred to as the \emph{decision statistics}. Consider decision statistics which satisfy the condition
\begin{itemize}
\item[C3] There is a function $u(\phi)$ such that $u(\phi(\theta,x)) = \tfrac{\pi(\theta)}{g(\theta)}$.
\end{itemize}
Then the optimal $\alpha$ amongst those of the form $\alpha(\theta,x)=\alpha(\phi(\theta,x))$ is as follows for some $\lambda \geq 0$,
\begin{equation} \label{eq:tuning2}
\alpha(\phi) = \min\left\{1, \lambda u(\phi) \left[\frac{\gamma(\phi)}{\bar{T}_2(\phi)}\right]^{1/2} \right\},
\end{equation}
where $\gamma(\phi) = \Pr \left(d(s(Y), s(y_{\text{obs}})) \leq \epsilon | \phi \right)$ and $\bar{T}_2(\phi)$ is as defined in Section \ref{sec:termnot}.
See Figures \ref{fig:SIR}D and \ref{fig:SErep}D for example $\alpha(\phi)$ functions based on estimating \eqref{eq:tuning2}.

Some cases in which condition C3 holds are:
$\pi(\theta)=g(\theta)$ (ABC rejection sampling); $\phi(\theta,x)=(\theta,\ldots)$ (the decision statistics include $\theta$); $\phi(\theta,x)=(\tfrac{\pi(\theta)}{g(\theta)}, \ldots)$.
The optimal $\alpha$ when this condition does not hold is described by Theorem 2 in Appendix C.

\subsubsection{Tuning $g$}

It is not clear what the optimal choice of $g(\theta)$ is for lazy ABC. The examples later use typical choices from the ABC literature, but a better choice may improve performance further.
This paper's theory gives the following corollary on the optimal choice of $g(\theta)$ for ABC-IS and RW-IS, which is of some general interest.

\begin{corollary} \label{cor:ABC-IS}
The asymptotic efficiency of ABC-IS is maximised by $g(\theta) \propto \pi(\theta) \left[\dfrac{\gamma(\theta)}{\bar{T}(\theta)}\right]^{1/2}$,
where $\gamma(\theta) = \E(\hat{L}_{\text{ABC}} | \theta)$.
This also holds for RW-IS with $\gamma(\theta) = \E(\hat{L}^2 | \theta)$.
\end{corollary}

\begin{proof} See Appendix \ref{sec:g}. \end{proof}


\subsection{Methods} \label{sec:tuning practice}

The theory above motivates choosing $\alpha$ by estimating \eqref{eq:tuning2}.  This section details a method to implement this approach.  Its effectiveness is discussed in Section \ref{sec:discussion}.

The method is outlined here and more detail is given in Sections \ref{sec:tbarest} to \ref{sec:movingepsilon}.
Tuning begins with a pilot run of $N'$ iterations of ABC-IS,
recording intermediate simulation states and timings of interest.
This is used to estimate $\gamma(\phi)$ and $\bar{T}_2(\phi)$ for various choices of $X$ and $\phi$, considering only $\phi$ such that condition C3 holds.  Under each of these choices, $\lambda$ is found by numerically maximising an estimate of efficiency.  The optimal choice of $X$ and $\phi$ is then made.
Following tuning, $N$ iterations of lazy ABC are performed (unless the estimated efficiency gains are judged inadequate, in which case ABC-IS can be used).

A simpler variation on this method is possible if the set of possible $\phi(\theta,x)$ values is finite and small. A pilot run is performed as before. Then $\alpha(\phi)$ values are chosen by direct numerical optimisation of an estimate of the algorithm's efficiency.



\subsubsection{Estimation of $\bar{T}_2(\phi)$} \label{sec:tbarest}

It may often suffice to treat $\bar{T}_2(\phi)$ as constant and estimate it as the mean CPU time of the continuation stage in the pilot run.  This is the case if knowledge of the simulation process shows the number of computational operations required is unaffected by $\phi$, or if the pilot run shows $\bar{T}_2(\phi)$ varies little relative to $\gamma(\phi)$.  Alternatively, statistical methods such as regression can be used for estimation, which is straightforward when $\phi$ is low dimensional.

\subsubsection{Estimation of $\gamma(\phi)$} \label{sec:gammaest}

Estimation of $\gamma(\phi)$ is more difficult.  Two approaches are suggested: the ``standard'' approach, producing $\hat{\gamma}^{(1)}$, attempts accurate estimation but typically involves strong assumptions; the ``conservative'' approach, producing $\hat{\gamma}^{(2)}$, sacrifices accuracy to improve robustness.  They are based on two equivalent expressions for $\gamma(\phi)$: $\Pr(d(s(Y),s(y_{\text{obs}})) \leq \epsilon | \phi)$ and $\E[\hat{L}_{\text{ABC}} | \phi]$.  Examples of successful implementations of both approaches are given in Sections \ref{sec:epidemics} and \ref{sec:extremes}.

The standard approach is to model the relationship between $\phi$ and $d(s(Y), s(y_{\text{obs}}))$ and use this to estimate $\Pr(d(s(Y),s(y_{\text{obs}})) \leq \epsilon | \phi)$.
However a difficulty is that for most $\phi$ values this involves extrapolating into the tails of the distribution of $d(s(Y),s(y_{\text{obs}})) | \phi$.  See Figure \ref{fig:SErep}A for example.  This creates a danger of underestimating the optimal $\alpha$ values and potentially producing very large importance sampling weights.
(This can be avoided for simple summary statistics by first modelling $s(Y)|\phi$; see Section \ref{sec:epidemics} for example.)

The conservative approach is to select $\epsilon_1$ following the pilot run such that a sufficiently large number of its simulations $y_{1:N'}$ satisfy $d(s(y_i),s(y_{\text{obs}})) \leq \epsilon_1$.  Let $z_i$ be indicator variables denoting meeting this condition and model the relationship between $z_i$ and the simulated $\phi_i$ values.  One method, used in the application later, is non-parametric logistic regression following \cite{Wood:2011}.  This approach is effectively tuning the method based on an $\epsilon$ value larger than that of interest.  This is an inefficient way to sample from the target of interest.  However, if tuning can be done well for the larger $\epsilon$ value then this method is safe from producing any dangerously large importance weights, as discussed further in Section \ref{sec:movingepsilon}.  Nonetheless, for $\phi$ regions where there are no $z_i=1$ values the conservative estimate is still based on extrapolation and unlikely to be accurate.  Consequences of this are discussed in Section \ref{sec:discussion}.

\subsubsection{Estimating efficiency} \label{sec:effest}

The tuning method outlined above requires the use of $N'$ pilot run iterations to estimate the efficiency of lazy ABC under various choices of tuning details (in particular $X$, $\phi$ and $\alpha$).  It is sufficient to estimate $[\E(W^2) \E(T)]^{-1}$, as this equals efficiency up to a constant of proportionality.  This can be used to estimate efficiency relative to ABC-IS, which is a particularly interpretable form of the results as it shows the efficiency improvement of using lazy ABC.

Assume that for a particular choice of tuning details the following are available for $1 \leq i \leq N'$: $t^{(1)}_i$ - initial simulation stage time; $t^{(2)}_i$ - continuation simulation stage time; $\alpha_i$ - continuation probability; $\hat{\gamma}_i$ - estimate of $\E(\hat{L}_{\text{ABC}} | \phi_i)$; $u_i$ - ratio $\pi(\theta)/g(\theta)$.  An estimate up to proportionality of efficiency is then $[\widehat{W^2} \hat{T}]^{-1}$ where 
$\widehat{W^2} = N'^{-1} \sum_{i=1}^{N'} u_i^2 \hat{\gamma}_i/\alpha_i$ and
$\hat{T} = \sum_{i=1}^{N'} t^{(1)}_i + \sum_{i=1}^{N'} \alpha_i t^{(2)}_i$.
An estimate of efficiency of ABC-IS is formed by taking $\alpha \equiv 1$.  Note that this often overestimates $\E(T)$ as interrupting the simulation to record intermediate states and timings reduces algorithm efficiency.
A more precise estimate would be possible using further pilot simulations without this interruption.

\subsubsection{Combining pilot and main run output} \label{sec:combining}

To make efficient use of the pilot run, it can be used in the final output as well as for tuning.  This is done by appending the pilot sequence of $(\theta,w)$ pairs to that from the main algorithm.  Loosely speaking, since each individual sequence targets the same distribution, so does the combined sequence.  More technically, it is straightforward to see that ABC versions of relations \eqref{eq:IS1} and \eqref{eq:IS2} are roughly true for the combined sequence when $N$ and $N'$ are large, and are exactly true as $N \to \infty$ regardless of $N'$. Also note that on appending the sequences, gains in efficiency are possible by multiplying the weights of one sequence by a constant, but this is not implemented here as little improvement was observed in the application later.

\subsubsection{Choice of $\epsilon$} \label{sec:movingepsilon}

In ABC-IS, an appropriate value $\epsilon$ is often unknown a priori and is instead chosen based on the simulated $d(s(Y),s(y_{\text{obs}}))$ values.  For lazy ABC in this situation one can use the pilot run to select a preliminary conservative choice of $\epsilon_1$ as in Section \ref{sec:gammaest}
and perform lazy ABC with $\epsilon=\epsilon_1$.  Alternative values of $\epsilon$ can then be investigated by updating the realised $\hat{L}_{\text{ABC}}$ values in the weight calculations.  For $\epsilon < \epsilon_1$ this simply reduces the number of non-zero weights.  However $\epsilon \gg \epsilon_1$ is not recommended as this may introduce large weights and destabilise the importance sampling approximation.

\section{Example: infectious diseases} \label{sec:epidemics}

This section illustrates a basic implementation of lazy ABC and the tuning approach of Section \ref{sec:tuning practice} on a simplified version of the standard Markovian SIR infectious disease model.
While ABC is not necessary for inference under this model \citep[see][]{Andersson:2000}, it is used for more complex variants \citep{McKinley:2009,Brooks:2014}.

Suppose the random variables $S(t), I(t), R(t)$ follow a discrete time Markov chain.
These represent the population size which is susceptible, infectious or recovered after $t$ transitions.
Two transitions are possible.
Let $M=S(0)+I(0)+R(0)$ be the total population size.
With probability $k \tfrac{R_0}{M} S(t) I(t)$, an infection occurs giving $S(t+1)=S(t) - 1$, $I(t+1)=I(t)+1$ and $R(t+1)=R(t)$.
With probability $k I(t)$, a recovery occurs giving $S(t+1)=S(t)$, $I(t+1)=I(t)-1$ and $R(t+1)=R(t)+1$.
The constant $k$ is chosen so that the probabilities sum to 1.
When $I(t)=0$ the chain terminates
and a simple random sample of size $100$ is taken from the population.
The observed data $y_{\text{obs}}$ is the number of these which are recovered.
The parameter of interest is the \emph{basic reproduction number} $R_0$,
which is given a prior of $\text{Gamma}(3,1)$.
It is assumed that $M=10^5$, $I(0)=10^3$ and $R(0)=0$, giving $S(0)=9.9 \times 10^4$.

Standard ABC can be implemented here simply by taking $s(y)=y$ and $d(s_1,s_2)=|s_1-s_2|$.
Standard ABC rejection sampling was performed on a simulated dataset with $y_{obs}=73$ using $N=10^4$ iterations and $\epsilon=1$.
These settings are also used for the three lazy ABC analyses that follow.

First a simple version of lazy ABC with ad-hoc tuning choices is presented.
This considers stopping at $t=1000$, so that $X=(S(1000),I(1000),R(1000))$.
The decision statistic $\phi(\theta,X)$ is $I(1000)$.
It seems reasonable to stop if there is evidence that the number of infectious is not growing.
So let $\alpha(\phi) = 0.1$ if $\phi \leq 1000$, and $\alpha(\phi)=1$ otherwise.

The second lazy ABC analysis uses the standard tuning method of Section \ref{sec:tuning practice}, with the same choice of $X$ and $\phi$ as above.
A pilot run of 1000 standard ABC simulations was performed.
Figures \ref{fig:SIR}A and B plot pilot run realisations of $\phi$ against $y$, the observation, and the continuation stage simulation time.
From this $\bar{T}_2(\phi)$, expected remaining simulation time given $\phi$, was estimated by fitting a non-parametric regression, using a log link to enforce positivity.
Also $\hat{p}(\phi)$, the expectation given $\phi$ of the proportion of the population who are recovered when the Markov chain terminates, was estimated by fitting a non-parametric binomial GLM.
This GLM model fits the pilot data well as illustrated by Figure \ref{fig:SIR}A,
which shows the resulting expected number recovered in the final subsample and a 95\% confidence interval.
All non-parametric regressions were fitted using the approach of \cite{Wood:2011}.
An estimate of the probability of ABC acceptance, $\hat{\gamma}(\phi)$ can be directly calculated from $\hat{p}(\phi)$ and is shown in Figure \ref{fig:SIR}C.
It is now possible to estimate $\alpha(\phi)$ for a given value of $\lambda$ from \eqref{eq:tuning2}.
For each $\lambda$ the efficiency estimate of Section \ref{sec:effest} can be calculated.
Numerical optimisation shows this is maximised for $\lambda=3.45$.
Figure \ref{fig:SIR}D shows the estimated optimal $\alpha$ function.

Finally lazy ABC using the conservative tuning approach is considered.
This proceeds as for the standard tuning approach just described, except for the estimation of $\gamma(\phi)$.
A variable $z$ is introduced to indicate whether the pilot run data has ABC distance less than or equal to $\epsilon_1=3$.
This variable equals 1 for 50 cases and 0 for the remainder.
Then $\gamma(\phi)$ is estimated by fitting a non-parametric logistic regression of $z$ on $\phi$, with the results are shown in Figure \ref{fig:SIR}C.
This estimate is more conservative than that above in that it always assigns a higher probability of ABC acceptance.
Numerical optimisation chooses $\lambda=2.54$ and the resulting $\alpha(\phi)$ function is shown in Figure \ref{fig:SIR}D.

Table \ref{tab:SIR} shows the results of the analyses described above.
Note that each analysis used the same sequence of random seeds so that their $(\theta, X, Y)$ simulations were the same.
Each analysis accepted the same 194 parameter values, except for lazy ABC with conservative tuning which accepted a subset of 177 of these.
Under lazy ABC with ad-hoc tuning all weights were 1.
With standard tuning 189 weights equalled 1 with the others below 3.
With conservative tuning 131 weights equalled 1, with the others below 4.
Estimates of posterior quantities from all methods were very similar.

\begin{figure*}
\begin{center}
\includegraphics{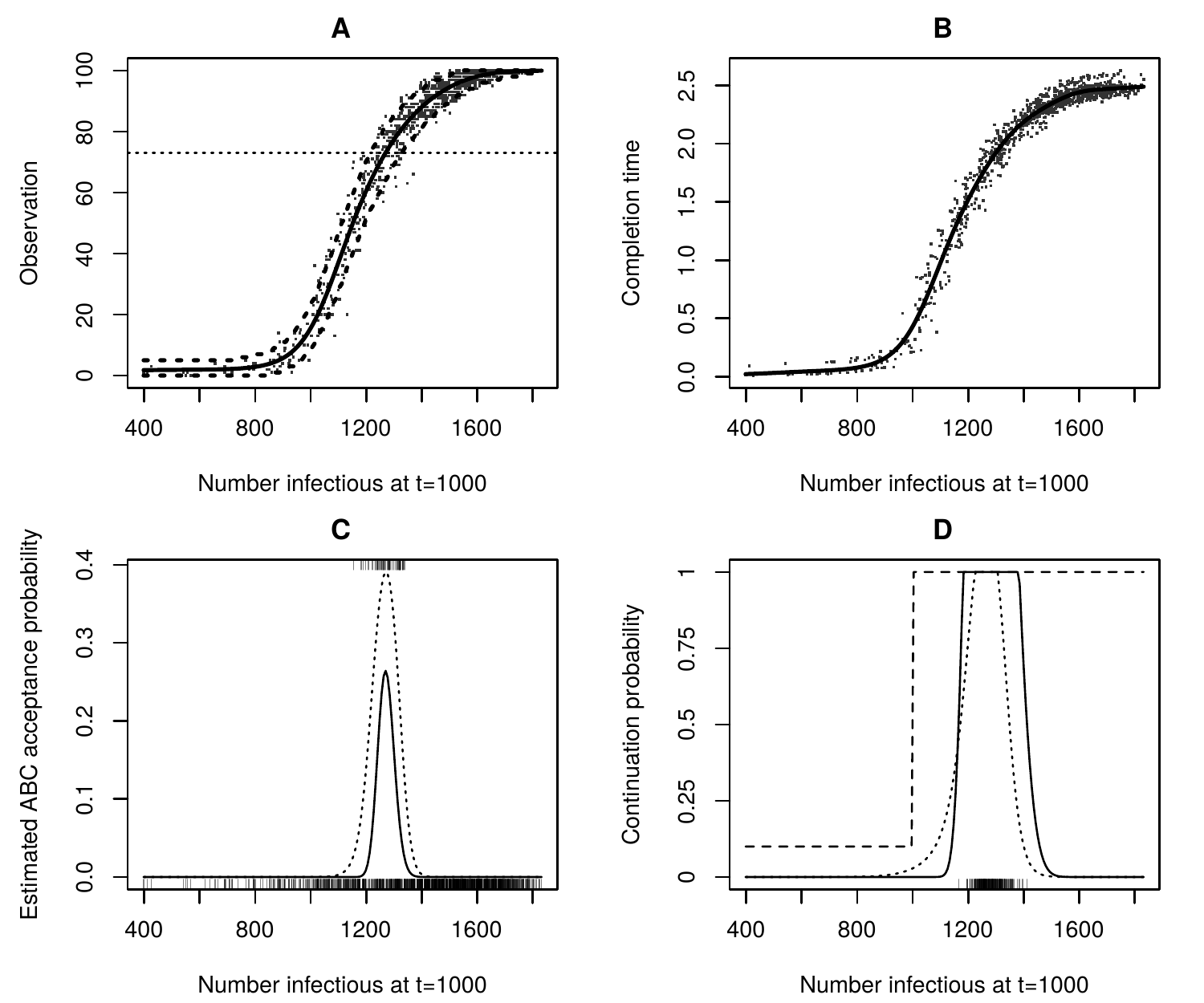}
\end{center}
\caption{\label{fig:SIR}
Tuning details of a simulation study applying lazy ABC to an SIR epidemic model.
\emph{Panel A} Pilot run realisation of $I(1000)$ (the decision statistic: number of infectious at $t=1000$) and $y$ (number recovered in subsample).
The solid line is the prediction under a non-parametric logistic regression and the dashed lines are corresponding 0.025 and 0.975 quantiles.
The dotted line shows the observed data.
\emph{Panel B} Pilot run realisations of $I(1000)$ and $T_2$ (continuation stage simulation time, in seconds).
The solid line is a fitted non-parametric regression.
\emph{Panel C} Estimated $\gamma(\phi)$ functions -- probability of eventual ABC acceptance -- under standard (solid) and conservative (dotted) tuning. Standard tuning derives this curve from panel A. Conservative tuning fits a logistic regression to the indicator variable represented by vertical marks. This indicates which pilot runs result in an ABC distance of 3 or less.
\emph{Panel D} $\alpha$ functions (continuation probabilities) under ad-hoc (dashed), standard (solid) and conservative (dotted) tuning. The latter two cases are derived from the pilot run as described in Section \ref{sec:tuning practice}.
The marks on the horizontal axis indicate the ABC simulations which were accepted.
}
\end{figure*}

\begin{table*}[htp]
\centering
\fbox{
\begin{tabular}{ccccc|cc}
\multirow{2}{*}{ABC method} & \multirow{2}{*}{Tuning} & \multirow{2}{*}{Time (s)} & \multirow{2}{*}{ESS} & Relative & \multicolumn{2}{c}{Estimated posterior} \\
           &              &          &       & efficiency & Mean & Standard deviation \\
\hline
Standard   & -            & 18124    & 194   & 1    & 1.803 & 0.1267 \\
Lazy       & Ad-hoc       & 17848    & 194   & 1.02 & 1.803 & 0.1267 \\
Lazy       & Standard     & 5113     & 192   & 3.51 & 1.804 & 0.1276 \\
Lazy       & Conservative & 3318     & 167   & 4.70 & 1.796 & 0.1212
\end{tabular}}
\caption{\label{tab:SIR}Simulation study on a SIR epidemic model. Four ABC analyses were performed on the same observations, sharing the same random seeds for their simulations. All used the same choice of $N=10^4$ and $\epsilon=1$. Iterations were run in parallel and computation times are summed over all cores used. Efficiency (ESS/CPU time) relative to standard ABC is shown. The pilot run required for tuning the last two rows tuning took a further 2016 seconds, with the remaining tuning calculations taking less than 2 seconds. For simplicity results are for lazy ABC output only, pilot run simulations have not been appended as in Section \ref{sec:combining}.}
\end{table*}

This example shows that it is straightforward to implement lazy ABC using ad-hoc tuning.
The time savings here are small as the $\alpha$ function is quite conservative, only allowing early stopping when the epidemic is already dying out at time $t=1000$.
A more effective $\alpha$ could be sought using results on the typical behaviour of this model \citep[see][]{Andersson:2000}.
The above instead considers the tuning method of Section \ref{sec:tuning practice} and shows it is simple to implement here and produces significant time savings.
Conservative tuning does better, but this may reflect variability of results for a relatively small number of iterations.

\section{Example: spatial extremes} \label{sec:extremes}

This section uses lazy ABC in a computationally demanding application of ABC to spatial extremes introduced by \cite{Erhardt:2012}. As well as illustrating the method for a complicated application, simulation studies are performed to investigate its efficiency.

\subsection{Background}

The observation $y_{t,d}$ represents the maximum measurement (e.g.~of rainfall) during year $t$ at location $x_d \in \mathbb{R}^2$.  There are $D$ locations and $M$ years.  The data are treated as $M$ independent replications of a spatial distribution.  Several models based on extreme value theory have been proposed, and \citeauthor{Erhardt:2012} concentrate on the \emph{Schlather process} \citep{Schlather:2002}.

The details of the Schlather process follow, together with the method \citeauthor{Erhardt:2012} used to calculate ABC summary statistics. The full mathematical details are not essential to understanding the implementation of lazy ABC. Instead, what is most important is the order of operations required to simulate a dataset and summary statistics. This is summarised as Algorithm \ref{alg:extremes sim} for later reference.

The Schlather process is based on independent identically distributed mean zero stationary Gaussian processes $U_i(x)$ where $i=1,2,\ldots$ (indexing an infinite number of Gaussian processes) and $x \in \mathbb{R}^2$ (representing location). 

The correlation between locations $x$ and $x'$ is given by the correlation function $\rho(h)$ where $h=||x-x'||_2$.  Let $s_i$ be draws from a Poisson process on $s \in (0,\infty)$ with intensity $\mu^{-1} s^{-2}$, where $\mu=\E[\max(0,U(x))]$.  Then the Schlather process is
\begin{equation} \label{eq:Schlather}
Y(x) = \max_i s_i \max(0, U_i(x)).
\end{equation}
\citeauthor{Erhardt:2012} focus on the Whittle-Mat\'ern correlation function with zero nugget
\begin{equation} \label{eq:Whittle-Matern}
\rho(h; c, \nu) = \frac{2^{1-v}}{\Gamma(\nu)} \left(\frac{h}{c}\right)^\nu K_\nu \left(\frac{h}{c}\right),
\end{equation}
where $\Gamma$ is the gamma function and $K_\nu$ is the modified Bessel function of the third kind with order $\nu$.  This has two parameters: range $c>0$ and smoothness $\nu > 0$.

A density function for the Schlather process is not available for $D>2$, making inference difficult.  \cite{Schlather:2002} provides a near-exact algorithm to simulate from the process based on only a finite number of copies of $U_i$,  motivating the use of ABC by \citeauthor{Erhardt:2012}.
They applied ABC rejection and importance sampling with a uniform prior on $[0,10]^2$ and investigated several choices of summary statistics.  The analysis here focuses on the choice they find most successful, based on \emph{tripletwise extremal coefficient estimators}.  Given a triple of 3 locations, $i,j,k$, this estimator is
\begin{equation} \label{eq:extremalcoeff}
\hat{\theta}_{ijk} = \frac{M}{\sum_{t=1}^{M}} 1/\max(y_{t,i}, y_{t,j}, y_{t,k}).
\end{equation}
There are $O(D^3)$ such summaries, so \citeauthor{Erhardt:2012} calculate a vector $m$ of mean values within $100$ clusters of triples, and use these as summary statistics.  Their clustering process 
finds triples of similar shapes, ignoring differences of location and rotation.  The ABC distance function between two vectors $m_1$ and $m_2$ of cluster means is
\begin{equation} \label{eq:sedist}
d(m_1, m_2) = \sum_{i=1}^{100} |m_{1i} - m_{2i}|.
\end{equation}
Although applying dimension reduction techniques to such high dimensional summaries has been shown to often improve ABC results \citep{Fearnhead:2012}, this is not investigated here as the aim is to investigate the efficiency improvements of lazy ABC.


Implementing the analysis below used the R packages ``SpatialExtremes'' \citep{Ribatet:2013} to simulate from the Schlather process and ``ABCExtremes'' to implement some details of the approach of \citeauthor{Erhardt:2012}.

\begin{algorithm}[htp]

\algline \\
{\bf Input:}
\begin{itemize}
\item Parameters: range $c>0$ and smoothness $\nu>0$.
\item Locations $x_1,\ldots,x_D$.
\item Number of years $M$.
\end{itemize}
{\bf Data simulation:}
\begin{enumerate}
\item[1] Calculate $D \times D$ covariance matrix $\Sigma(c,\nu)$ using \eqref{eq:Whittle-Matern}.
\item[2] Repeat the following steps for $t=1,\ldots,M$.
\begin{enumerate}
\item[a] Sample values $s_1, s_2, \ldots, s_q$ from a Poisson process (see text for intensity).
\item[b] For each $1\leq i \leq q$, sample $u_i(x_1), u_i(x_2), \ldots, u_i(x_D)$ from $N(0,\Sigma(c,\nu))$.
\item[c] Calculate Schlather process realisations for year $t$ by \eqref{eq:Schlather}.
\end{enumerate}
\end{enumerate}
{\bf Summary statistic calculation:}
\begin{enumerate}
\item[1] Calculate extremal coefficient estimate $\hat{\theta}_{ijk}$ for all triples of locations $ijk$ using \eqref{eq:extremalcoeff}.
\item[2] Calculate mean $\hat{\theta}_{ijk}$ values for each cluster of triples.
\end{enumerate}
{\bf Output:}
\begin{itemize}
\item Summary statistic vector comprising the cluster means.
\end{itemize}
\algline
\caption{Overview of simulation of Schlather process data and summary statistics. Full details of the steps are given in the text. \label{alg:extremes sim}}
\end{algorithm}

\subsection{Methods} \label{sec:simmethods}

Exploratory investigation of ABC code with $D=20$ and $M=100$ showed that the majority of time was spent simulating the data (7.1ms/iteration) and calculating extremal coefficient estimates (17.9ms/iteration), with the remaining steps being brief (3.1ms/iteration).  The time costs of the first two of these scaled with $D$ as roughly proportional to $D$ and $D^3$ respectively, so the latter is expected to dominate for large $D$.  Furthermore, interrupting and then resuming operations during the calculation of extremal coefficients is much simpler to implement than during simulation of data.  Therefore the initial simulation stage of the lazy ABC analysis was chosen to be simulating the data at all locations, and extremal coefficient estimates at a subset of locations $L$.  The continuation simulation stage was to calculate the remaining extremal coefficient estimates.

The decision statistic $\hat{d}$ was constructed as follows.  Let $m_{1i}$ be the $i$th cluster mean for the observed data.  Let $\hat{m}_{2i}$ be the $i$th cluster mean for the simulated data using only extremal coefficient estimates available at the initial simulation stage, and $B$ be the set of clusters for which any such estimates are available.  Then define $\hat{d} = \sum_{i \in B} |m_{1i} - \hat{m}_{2i}|$.  This is an estimate of the ABC distance $d$ \eqref{eq:sedist}.  It could be improved by estimating typical $\hat{m}_{2i}$ values for $i \not \in B$ but including such constant terms has no effect on the analysis below.

It was assumed that $\bar{T}_2(\hat{d})$ is constant as, given $D$, $L$ and $M$, the continuation stage always involves the same number of calculations.  The value was estimated by the mean CPU time for this stage in the pilot run.  
Analyses were performed using both the standard and conservative $\gamma$ estimators.
To calculate $\hat{\gamma}_1(\hat{d})$ from pilot run output, the relationship between $\hat{d}$ and $d$ was modelled statistically.  Exploratory analysis showed that there was a roughly linear relationship, but for some choices of $L$ this was heteroskedastic (see Figure \ref{fig:SErep}A).  Furthermore, for small $\hat{d}$ the distribution of $d | \hat{d}$ was skewed.  So $\Pr(d \leq \epsilon | \hat{d})$ was estimated based on a linear regression of $d$ on $\hat{d}$ with a Box-Cox transformation, using only simulations with nearby values of $\hat{d}$.  This was done for several $\hat{d}$ values and interpolated estimates elsewhere formed  $\hat{\gamma}_1(\hat{d})$.
For the importance sampling case, $\log u$ was also included in each regression giving a number of functions mapping $u$ to estimates of $\Pr(d \leq \epsilon | \hat{d}, u)$ for various $\hat{d}$ values, which were used for interpolation.
Calculation of $\hat{\gamma}_2$ was as described in Section \ref{sec:gammaest}, taking $\epsilon_1$ to give 100 acceptances in the pilot run.
Given estimates of $\bar{T}_2$ and $\gamma$, tuning was performed as described in Section \ref{sec:tuning practice}, with optimisation over possible choices of $L$ by backwards selection.

Three simulation studies were performed.
The first replicated the rejection sampling analysis of \citeauthor{Erhardt:2012} on several simulated datasets.  These used $D=20, M=100$ and true parameter values shown in Table \ref{tab:SErep}.  Each dataset used a different set of observation locations with integer coordinates sampled from $[0,10]^2$.  The first analysis was a replication of the standard ABC analysis, using $\epsilon$ values corresponding to 200 acceptances.  Then lazy ABC was performed on the same datasets under each method of estimating $\gamma$.  To compare the methods fairly, lazy ABC used the same $\epsilon$ value as standard ABC and reused its random seeds so that the sequence of $(\theta, X, Y)$ realisations is also the same.

The second simulation study investigated rejection sampling for a single larger simulated dataset with $D=35$, $c=0.5$ and $\nu=1$.  Locations were chosen as before.  As in a real application $\epsilon$ was not assumed to be known in advance and the approach of Section \ref{sec:movingepsilon} was used to select this post-hoc.  A complication for this dataset was that the simulation of Gaussian processes was difficult when both parameters were large: the default ``direct method'', based on Choleski decomposition, sometimes produced numerical errors.  Simulation was possible via the turning bands method (TBM) but much slower (roughly 150 times the CPU time).  A two stage simulation method was implemented.  First the direct method was attempted and if this failed TBM was used.  To save time lazy ABC was implemented with multiple stopping decisions, the first taking place after attempting the direct method.  This has a binary decision statistic indicating success or failure.  The second stopping decision is as described earlier.  Tuning was performed as described in Appendix \ref{sec:onecont}, using $\hat{\gamma}_1$ fitted as described above by either the standard or conservative tuning method.  The standard method used $\epsilon_1$ to give 30 acceptances in the pilot run.  As before all analyses reused the same random seeds.

Finally an importance sampling analysis was performed on the larger dataset.  A sample of $10^4$ log parameter values was taken from simulations of the preceding standard ABC analysis with distances below the 0.3 quantile.  A Gaussian mixture distribution was constructed with locations given by this sample and variances equal to twice the empirical variance of the sample.  After truncation to the prior support, this was used to give $g(\theta)$, where $\theta$ now represents the log parameters.  This choice follows the suggestions of \cite{Beaumont:2009},
noting that using the log scale produced a better fit to the sample and that the subsample was used to avoid slow density calculations.  The preceding $D=35$ analysis was then repeated.
As discussed in Section \ref{sec:tuning practice}, $u=\pi(\theta)/g(\theta)$ was included as a decision statistic.  Estimation of $\gamma$ and $\bar{T}_2$ was performed as before with $u$ included in the $\gamma$ estimate as described earlier.

All ABC analyses performed $10^6$ total iterations.  For lazy ABC $10^4$ of these comprised the pilot run.

\subsection{Results}

Figure \ref{fig:SErep} illustrates some details of tuning for one case of the $D=20$ study.  The results are shown in Table \ref{tab:SErep}.  For all datasets lazy ABC is roughly 4 times more efficient under the standard tuning method and 3 times under the conservative method.  
Efficiency gains for conservative tuning are slightly less than estimated.  This is because the estimate is made for a choice of $\epsilon_1$ larger than the final $\epsilon$.
The mean weights were also investigated, as these are useful in model selection as an estimate of $\pi(y_{\text{obs}})$.  All lazy ABC estimates differed from the standard ABC estimate by no more than 4\%.

\begin{figure*}
\begin{center}
\includegraphics{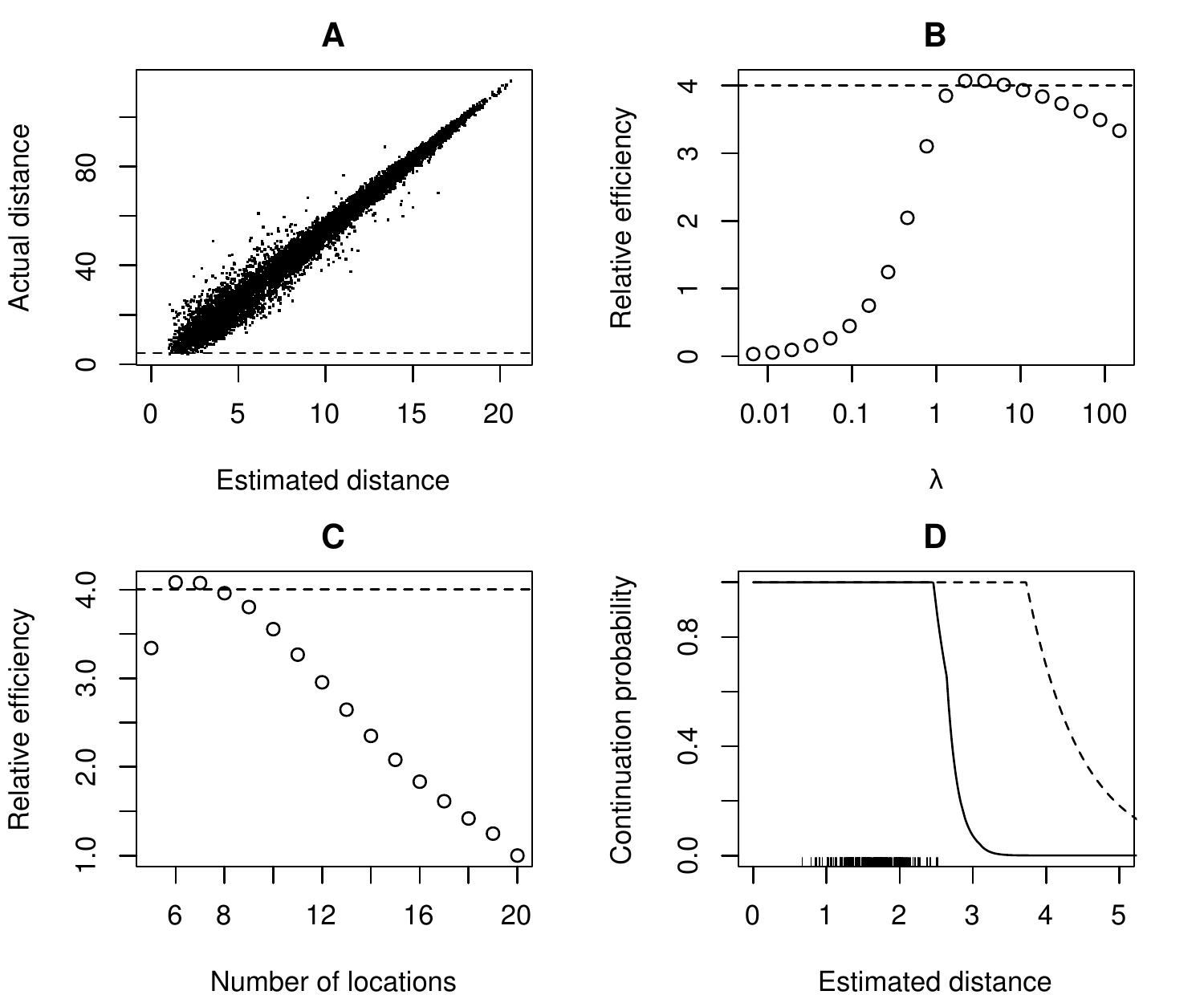}
\end{center}
\caption{\label{fig:SErep} Details of tuning lazy ABC in a simulation study on spatial extremes corresponding to the first row of Table \ref{tab:SErep}.  Panels A-C concentrate on the standard tuning approach.  \emph{Panel A} Pilot run values of $\hat{d}$ (estimated distance) and $d$ (actual distance).  The dashed line shows the value of $\epsilon$.  \emph{Panel B} Estimates of efficiency (ESS/time) for different values of $\lambda$ (parameter of \eqref{eq:tuning2} which must be determined numerically).  The dashed line shows the realised efficiency (ESS/time of the final lazy ABC algorithm).  \emph{Panel C} Estimated efficiency for the best choices of $L$ (subset of locations from which early stopping is considered) of various lengths output by backwards selection.  The dashed line shows the realised efficiency.  \emph{Panel D} Values of $\hat{d}$ and $\alpha$ (continuation probability) from non-pilot simulations under standard (solid line) and conservative (dashed line) tuning.  The marks on the horizontal axis indicate the simulations which resulted in positive weights.  (For this panel conservative tuning was performed using $L$ as selected by standard tuning.)}
\end{figure*}

\begin{table*}[pht]
\centering
\fbox{ \footnotesize
\begin{tabular}{cc|c|ccc|cc}
      &        & Standard       & \multicolumn{3}{|c|}{Lazy} & \multicolumn{2}{|c}{Relative efficiency}                \\
Range & Smooth & Time ($10^3$s) & Time ($10^3$s)             & Sample size & ESS           & Estimated   & Actual      \\ 
  \hline                                        
 0.5  & 1      & 32.0           & 8.0  (11.6)                & 196 (199)   & 196.0 (198.7) & 4.08 (3.28) & 4.00 (2.79) \\ 
 1    & 1      & 31.3           & 7.3  (9.8)                 & 200 (200)   & 199.9 (200.0) & 4.34 (4.31) & 4.35 (3.25) \\ 
 1    & 3      & 31.3           & 8.2  (11.2)                & 194 (198)   & 182.5 (196.5) & 3.77 (3.43) & 3.51 (2.79) \\ 
 3    & 1      & 31.2           & 7.7  (11.1)                & 194 (200)   & 189.9 (200.0) & 4.18 (3.56) & 3.89 (2.86) \\ 
 3    & 3      & 31.2           & 7.4  (11.0)                & 192 (199)   & 175.8 (199.0) & 4.43 (3.65) & 3.79 (2.87) \\ 
 5    & 3      & 31.3           & 8.3  (11.1)                & 200 (200)   & 200.0 (200.0) & 3.73 (3.49) & 3.85 (2.87)
\end{tabular}}
\caption{\label{tab:SErep}Simulation study on spatial extremes replicating \cite{Erhardt:2012}.  Each row represents the analysis of a simulated dataset under the given values of range $c$ and smoothness $\nu$ parameters (used in \eqref{eq:Whittle-Matern}). In each analysis a choice of $\epsilon$ was made under standard ABC so that the accepted sample size (and therefore ESS) was 200, and the same value was used for lazy ABC.  Lazy ABC figures are shown for both the standard $\hat{\gamma}$ estimate and, in brackets, the conservative estimate.  The lazy ABC output includes the pilot run as described in Section \ref{sec:combining}, and also includes the tuning time (roughly 120 seconds for the standard approach and 210 for the conservative).  Iterations were run in parallel and computation times are summed over all cores used.  For all datasets efficiency (ESS/time) to 1 significant figure was 0.006 for standard ABC and 0.02 or 0.03 under either approach to lazy ABC.}
\end{table*}

Table \ref{tab:SEbig} shows results for the $D=35$ dataset.  In the initial rejection sampling analysis lazy ABC improved efficiency by roughly 8 times.  For importance sampling the improvement factor is 2, showing that lazy ABC still improves efficiency, although this is harder when $g$ concentrates on plausible choices of $\theta$.  For example, standard ABC now spends negligible time on TBM simulations.  As before lazy ABC estimates of $\pi(y_{\text{obs}})$ differed from those of standard ABC by no more than 4\%.  In both cases a post-hoc selection of $\epsilon$ has been used successfully.

Table \ref{tab:SEbig} shows that under rejection sampling the standard tuning ESS is considerably smaller than the sample size.  This is due to two simulations which are given importance weights of 10.  These have $\hat{\gamma}$ values of roughly $10^{-4}$ which appears to be an underestimate.  Conservative tuning avoids large importance weights to give an ESS of roughly 200 and improves the relative efficiency for the same final choice of $\epsilon$.  For importance sampling conservative tuning also performs better.  The reason here is not obvious but may be a better final selection of $L$.

\begin{table*}[pht]
\centering
\fbox{ \footnotesize
\begin{tabular}{l|c|ccc|ccc|c}
                      &            & \multicolumn{3}{|c|}{Standard} & \multicolumn{3}{|c|}{Lazy} & Relative                                                        \\
                      & $\epsilon$ & Time ($10^3$s)                 & Sample size                & ESS  & Time ($10^3$s)        & Sample size & ESS   & efficiency \\
\hline                                        
RS standard           & 2.61       & 241.4                          & 210                        & 210  & 22.8                  & 200         & 139.6 & 7.0        \\
RS conservative       & 2.61       & 241.4                          & 207                        & 207  & 25.5                  & 200         & 197.7 & 9.0        \\
\hline
IS standard           & 2.33       & 136.4                          & 209                        & 168  & 61.9                  & 200         & 165   & 2.2        \\
IS conservative       & 2.33       & 136.4                          & 209                        & 168  & 51.0                  & 200         & 162   & 2.6        \\ 

\end{tabular}}
\caption{\label{tab:SEbig}Simulation study on a spatial extremes dataset with $D=35$.  Results are shown for rejection and importance sampling with standard and conservative tuning.  The rejection sampling output was used to create the importance density.  The final choice of $\epsilon$ is shown.  For IS the two $\epsilon$ values are equal but there is a small difference for RS.   The lazy ABC output includes the pilot run and the tuning time.}
\end{table*}

\section{Discussion} \label{sec:discussion}

This paper has introduced lazy ABC, a method to speed up inference by ABC importance sampling without introducing further approximations.  The approach is to abandon some unpromising simulations before they are complete.  By using a probabilistic stopping rule and weighting the accepted simulations accordingly, the algorithm targets exactly the same distribution as standard ABC, in the sense that Monte Carlo estimates of functions $h(\theta)$ and of the model evidence converge to unchanged values.

Results have been provided on the optimal tuning of the lazy ABC stopping rule and used to motivate a practical tuning method.  This has been demonstrated for a simple epidemiological example and a computationally challenging application where it has produced improvements in efficiency (ESS/CPU time) over standard ABC of up to 8 times.

The tuning method is based on estimating the optimal choice of $\alpha(\phi)$, \eqref{eq:tuning2}.  The most difficult part was estimating $\gamma(\phi) = \Pr \left(d(s(Y), s(y_{\text{obs}})) \leq \epsilon | \phi \right)$ from pilot run data.  Two approaches to this were described, a standard approach of direct estimation and a conservative approach of estimation using a larger $\epsilon$ value than is of interest for ABC.  The latter approach improves robustness and make estimation simpler at the cost of some inefficiency.  Both approaches performed well in the simulation studies but some improvements are desirable.  Firstly, estimation of $\gamma(\phi)$ often involves extrapolation which may produce inaccurate results.  Secondly, several choices by the user are required, especially for the standard approach.  A more automated approach would be useful for lazy versions of ABC SMC algorithms, where a new choice of $\alpha$ would be needed for each $\epsilon$ value,  or alternatively for lazy ABC algorithms which adapt $\alpha$ as more simulations become available.  It would be of interest to find suboptimal but robust choices of $\alpha$ addressing these issues.

A related point is that lazy ABC can be generalised to allow a non-uniform ABC kernel.
This may allow alternative approaches to tuning.

In some situations likelihood-based inference is possible, but calculating the likelihood or an unbiased estimate is expensive.
Generalising lazy ABC ideas to this situation is possible.
Appendix \ref{sec:LIS} describes this \emph{lazy importance sampling} algorithm and shows that the theoretical results of the paper carry over to it.
It also discusses why practical application seems more challenging than for lazy ABC.
Nonetheless this is an interesting topic for future research.

Lazy ABC with multiple stopping decisions is another extension to the framework of the main paper and is described in Appendix \ref{sec:multistop},
with an example of implementation in Section \ref{sec:extremes}.
A tuning method is given when the decision statistics for all stopping decisions are discrete, and also some cases where one decision statistic is continuous.  For more complex cases tuning results are not available.  For now it is recommended to discretise most decision statistics to avoid this difficulty.


This paper has concentrated on importance sampling, which is widely used by ABC practitioners, but the lazy ABC approach can be extended to ABC versions of MCMC and SMC, which are more efficient algorithms.  The tuning results are applicable to SMC algorithms, but further practical methods are needed, as mentioned above.  Further theory on optimal tuning is necessary for MCMC, although good performance may be possible with ad-hoc tuning.
It would also be of interest
to design algorithms in which the initial simulation stages can be resampled and continued many times.

\paragraph{Acknowledgements} 
Thanks to Chris Sherlock, Richard Everitt, Scott Sisson, Christian Robert and two anonymous referees for helpful suggestions, and Robert Erhardt for advice on the spatial extremes example.

\appendix

\section{Lazy importance sampling} \label{sec:LIS}

The approach of lazy ABC can be generalised to non-ABC situations to give \emph{lazy importance sampling} (LIS).  This is Algorithm \ref{alg:RW-IS} using a likelihood estimator of the form:
\[
\hat{L}_{\text{lazy}} = \begin{cases} 
  \hat{L} / \alpha(\theta, X) & \text{with probability } \alpha(\theta, X) \\
  0 & \text{otherwise}
\end{cases}
\]
In addition to condition C1 from Section \ref{sec:ESalgorithm} assume:
\begin{itemize}
\item[C4] The distribution $(X,\hat{L}) | \theta$ is such that $\hat{L} | \theta$ is a non-negative unbiased estimator of $L(\theta)$, and both $X | \theta$ and $\hat{L} | \theta, x$ can be simulated from.
\end{itemize}
This framework can be used when $\hat{L}$ is an expensive unbiased estimator.  It also allows cases where either or both of $X$ and $\hat{L}$ are non-random.  For example, $X$ may be a deterministic approximation of the likelihood and $\hat{L} | \theta$ may be a point mass at $L(\theta)$.

Close analogues of the lazy ABC results in this paper hold for LIS.
Firstly, a variant of Theorem \ref{thm:es} shows that given conditions C1 and C4, $\hat{L}_{\text{lazy}} | \theta$ is a non-negative unbiased estimator of $L(\theta)$.
The same proof can be used replacing $L_{\text{ABC}}(\theta)$ and $\hat{L}_{\text{ABC}}$ with $L(\theta)$ and $\hat{L}$.
It follows that LIS targets the posterior distribution.

Secondly, modified versions of equations \eqref{eq:tuning1} and \eqref{eq:tuning2} on the optimal choice of $\alpha$ hold as before.
In particular, condition C3 is still required for \eqref{eq:tuning2}.
The modification is that under LIS $\gamma(\phi) = \E[ \hat{L}^2 | \phi ]$,
whereas under lazy ABC $\gamma(\phi)$ can be written as $\E[ \hat{L}_{\text{ABC}} | \phi ]$.
The lazy ABC and LIS results are both proved by Theorem \ref{thm:tuning} in Appendix \ref{sec:alpha}.

Exploratory investigation suggests that tuning LIS is harder in practice than lazy ABC.
This is because $\E(\hat{L}^2 | \phi)$ can be strongly influenced by the upper tail of $\hat{L} | \phi$ which is hard to estimate from pilot run output.
Furthermore it is unclear whether the potential gains of LIS are comparable to lazy ABC.
Under ABC as $\epsilon \to 0$ the acceptance rate typically converges to zero, so for small $\epsilon$ there is scope to save time by early stopping since most iterations do not contribute to the final sample.
It is unclear whether this is the case under importance sampling with a good choice of $g(\theta)$.

Further work to investigate the usefulness of LIS as an inference method is required.
LIS is also included because it is a useful device for simplifying the proofs in the remaining appendices.

\section{Multiple stopping decisions} \label{sec:multistop}

The lazy ABC framework of Section \ref{sec:ESalgorithm} allows multiple stopping decisions, as follows.
As in that section assume $Y$ is a deterministic transformation of a latent vector $X_{1:p}$.
\begin{itemize}[leftmargin=*]
\item[] \emph{Example B1: Multiple stopping decisions} Let $X=X_{1:p}$ and $\alpha(\theta, x)=\prod_{i=1}^s \alpha^{(i)}(\theta, x_{1:t_i})$.  Thus, for each $1 \leq i \leq s$, once simulation of $X_{1:t_i}$ has been performed then $\hat{L}_{\text{lazy}}$ is set to zero with a certain probability, in which case no further simulation is necessary.  It is often be useful to let $\alpha^{(i)}(\theta, x_{1:t_i}) = \alpha^{(i)}(\phi_i(\theta,x_{1:t_i}))$.  That is, each stopping decision has associated decision statistics $\phi_i$.
\item[] \emph{Example B2: Multiple random stopping times} As for Example B1 but with each $t_i$ replaced with a random stopping time variable $T_i$.  This permits stopping to be considered when various random events occur, without imposing a fixed order of occurrence.
\end{itemize}
The following alternative characterisation of these examples is useful below.

\begin{lemma} \label{lem:multistop}
For any $1 \leq i \leq s$, Examples B1 and B2 can be represented as a lazy importance sampling algorithm with continuation probability $\alpha^{(i)}(\phi_i)$ and 
\[
\hat{L}=
\begin{cases}
\hat{L}_{\text{ABC}} / \beta_i(\theta, X) & \text{with probability } \beta_i(\theta,X) \\
0 & \text{otherwise},
\end{cases}
\]
where $\beta_i(\theta, x) = \prod_{j \neq i} \alpha^{(j)}(\phi_j)$.
\end{lemma}

\begin{proof}
The likelihood estimator stated can easily be verified to have the same distribution as $\hat{L}_{\text{lazy}}$.
\end{proof}

It is also helpful to define $\bar{T}_{2i}(\theta, \phi_{1:s})$ as the expected time remaining from the calculation of $\phi_i$ until the likelihood estimate is computed conditional on $\theta$ and $\phi_{1:s}$, and $\bar{T}_{2i}(\phi_i) = \E[\bar{T}_{2i}(\theta, \phi_{1:s}) | \phi_i]$.

\subsection{Tuning} \label{sec:multistoptuning}

The efficiency estimate of Section \ref{sec:effest} can be used in a multiple stopping decision setting given a choice of $\alpha$.  It is necessary to update the estimator $\hat{T}$ given there which is usually a straightforward task.  Sections \ref{sec:alldiscrete} and \ref{sec:onecont} describe situations of practical interest where the optimal form of $\alpha$ can be derived.  However in general the problem is challenging, as illustrated by Section \ref{sec:twocont}.

\subsubsection{Discrete decision statistics} \label{sec:alldiscrete}

Suppose $\alpha(\theta, x) = \prod_{i=1}^s \alpha^{(i)}(\phi_i)$ where $\phi_i(\theta, x)$ takes values in $\{1, 2, \ldots, d_i\}$ for $d_i$ finite.  Tuning requires selecting a finite number of $\alpha^{(i)}(\phi_i)$ values to optimise the efficiency estimate, which is possible by standard numerical optimisation methods.  However note that producing an efficiency estimate as in Section \ref{sec:effest} becomes difficult for large $s$.

\subsubsection{One continuous decision statistic} \label{sec:onecont}

Suppose $\alpha(\theta, x) = \prod_{i=1}^s \alpha^{(i)}(\phi_i)$ where $\phi_1(\theta, x)$ is continuous and $\phi_i(\theta, x)$ is as in Section \ref{sec:alldiscrete} for $i>1$.
Also suppose there exists $u_1(\phi_1)=\pi(\theta)/g(\theta)$, so that condition C3 holds.
Applying Lemma \ref{lem:multistop} and \eqref{eq:tuning2} with $\gamma(\phi) = \E[ \hat{L}^2 | \phi ]$, as justified in Appendix \ref{sec:LIS}, gives that efficiency is optimised by
\begin{equation} \label{eq:alpha1}
\alpha^{(1)}(\phi_1) = \min\left\{1, \lambda u_1(\phi_1) \left[ \frac{\gamma_1(\phi_1)}{\bar{T}_{21}(\phi_1)} \right]^{1/2} \right\},
\end{equation}
where $\gamma_i(\phi_1) = \E\left[\zeta(\beta_i(\theta, X)) \mathbbm{1} \left\{ d(s(Y),s(y_{\text{obs}}))\leq \epsilon \right\} | \phi_i \right]$, $\zeta(0)=0$ and $\zeta(x)=x^{-1}$ for $x>0$. 

In general $\gamma_1(\phi_1)$ and $\bar{T}_{21}(\phi_1)$ depend on $\alpha^{(i)}(\phi_i)$ for $i>2$ and so must be estimated several times during the tuning process which is costly.  A special case where this can be avoided is when $\phi_{2:p}$ is fully determined by $\phi_1$ (and so typically the decision associated with $\alpha_1$ is guaranteed to occur last).  For example this is the situation in the second simulation study of Section \ref{sec:simmethods}.


\subsubsection{Multiple continuous decision statistics} \label{sec:twocont}

Consider the setting of \ref{sec:onecont} with the modification that every $\phi_i(\theta, x)$ is continuous and there exists a corresponding function $u_i(\phi_i)=\pi(\theta)/g(\theta)$.  The same approach as above gives equations of the form
\[
\alpha^{(i)}(\phi_i) = \min\left\{1, \lambda_i u_i(\phi_i) \left[ \frac{\gamma_i(\phi_i)}{\bar{T}_{2i}(\phi_i)} \right]^{1/2} \right\},
\]
for $i=1,\ldots,s$.  The definition of $\gamma_i$ involves $\alpha^{(j)}$ for all $j \neq i$, and $\bar{T}_{2i}$ will also involve many of these terms.  Thus deriving the optimal $\alpha^{(i)}$ functions involves solving a complicated system of non-linear implicit equations.

\section{Tuning $\alpha$} \label{sec:alpha}

This appendix concerns tuning $\alpha$ to optimise the asymptotic efficiency of lazy importance sampling. Theorem \ref{thm:tuning} on this topic is stated and several earlier results are shown to be corrolaries. The proof is given in Appendix \ref{sec:alphaproof}.

Recall from Section \ref{sec:tuning theory} that efficiency is defined here as $N_{\text{eff}} / T$ where $T$ is the CPU time of the algorithm and $N_{\text{eff}} = N E(W)^2 / E(W^2)$. In the latter $W$ represents importance sampling weight and $N$ number of iterations. Assume that $T$ follows a central limit theorem in $N$.  Then the delta method gives that for large $N$ efficiency asymptotically equals
\begin{equation} \label{eq:asymptotic efficiency}
\frac{\E(W)^2 / \E(W^2)}{\E(T)/N}.
\end{equation}

\begin{theorem} \label{thm:tuning}
Fix some decision statistics $\phi(\theta,x)$.  Amongst continuation probability functions of the form $\alpha(\theta,x)=\alpha(\phi(\theta,x))$, asymptotic efficiency \eqref{eq:asymptotic efficiency} is maximised by the following expression for some $\lambda>0$,
\begin{equation} \label{eq:tuning}
\alpha(\phi) = \min\left\{1, \lambda \left[\frac{\E[\hat{L}^2 \tfrac{\pi(\theta)^2}{g(\theta)^2} | \phi]}{\bar{T}_2(\phi)}\right]^{1/2} \right\}.
\end{equation}
\end{theorem}

Note that the expectation in the numerator is conditional on $\phi$, with $\theta$ and $x$, and $y$ in the lazy ABC case described below, treated as random.

If condition C3 holds then \eqref{eq:tuning} simplifies to \eqref{eq:tuning2} with $\gamma(\theta) = \E(\hat{L}^2 | \theta)$.
If $\phi(\theta,x)=(\theta,x)$ it simplifies further to \eqref{eq:tuning1}.
This proves the results stated in Appendix \ref{sec:LIS}.

To apply Theorem \ref{thm:tuning} to lazy ABC let $\hat{L} = \hat{L}_{\text{ABC}}$ (as defined in \eqref{eq:LhatABC}).
Then LIS is equivalent to the lazy ABC algorithm.
Note that since $\hat{L}_{\text{ABC}}$ is a Bernoulli random variable, $\hat{L}^2=\hat{L}_{\text{ABC}}$.
If condition C3 holds then \eqref{eq:tuning} simplifies to \eqref{eq:tuning2} with
$\gamma(\theta) = \E(\hat{L}_{\text{ABC}} | \theta) = \Pr \left(d(s(Y), s(y_{\text{obs}})) \leq \epsilon | \theta, x \right)$.
If $\phi(\theta,x)=(\theta,x)$ it simplifies further to \eqref{eq:tuning1}.
This proves the results stated in Section \ref{sec:tuning}.

\subsection{Proof of Theorem \ref{thm:tuning}} \label{sec:alphaproof}

In LIS the importance sampling weight $W$ equals $\frac{\hat{L} \pi(\theta)}{\alpha(\phi) g(\theta)}$ with probability $\alpha(\phi)$ and zero otherwise.
Hence:
\begin{equation} \label{eq:W2}
\E(W^2) = \int \frac{\E[\hat{L}^2 | \theta, \phi, y] \pi(\theta)^2}{\alpha(\phi) g(\theta)^2} \pi(\phi,y|\theta) g(\theta) d\theta d\phi dy
       = \int \frac{\xi(\phi)}{\alpha(\phi)} g(\phi) d\phi,
\end{equation}
where $\xi(\phi) = \E\left[ \hat{L}^2 \left( \tfrac{\pi(\theta)}{g(\theta)} \right)^2 \middle| \phi \right]$
and $g(\phi) = \int \pi(\phi|\theta) g(\theta) d \theta$.

The expected time of a single iteration of the LIS algorithm is
\begin{equation} \label{eq:ET}
\E(T)/N = \bar{T}_1 + \int \alpha(\phi) \bar{T}_2(\theta, \phi) \pi(\phi|\theta) g(\theta) d\theta d\phi
= \bar{T}_1 + \int \alpha(\phi) \bar{T}_2(\phi) g(\phi) d\phi.
\end{equation}

Note that $\E(W)$ is a constant, so choosing $\alpha(\phi)$ to maximise expression \eqref{eq:asymptotic efficiency} is equivalent to minimising $\E(W^2) \E(T)/N$.  Call this problem $P$.
Consider also the problems $P(\upsilon)$, minimising $\E(W^2)$ under the constraint $\E(T)/N = \upsilon$ and $P(\upsilon, \mu)$, minimising $\E(W^2) + \mu [E(T)/N - \upsilon]$, or equivalently 
\begin{equation} \label{eq:LM}
\int \left[ \frac{\xi(\phi)}{\alpha(\phi)}
+ \mu \alpha(\phi) \bar{T}_2(\phi) \right] g(\phi) d\phi.
\end{equation}
Note that $P(\upsilon, \mu)$ is a Lagrange multiplier form of $P(\upsilon)$.  Consider only $\mu>0$.  Also, let $\Upsilon$ be the set of $\E(T)/N$ values attainable by some choice of $\alpha$.

First consider minimising \eqref{eq:LM} subject to $0 \leq \alpha(\phi) \leq 1$.
This can be done by pointwise optimisation of the integrand.  With $\alpha$ unconstrained the solution is
\begin{equation} \label{eq:alphastar}
\alpha^*(\phi) = \lambda \left[\frac{\xi(\phi)}{\bar{T}_2(\phi)}\right]^{1/2},
\end{equation}
where $\lambda = \mu^{-1/2}$.  Also note that $\alpha^*(\phi)$ may sometimes be infinite.  The derivative of the integrand with respect to $\alpha$ is negative for $\alpha < \alpha^*$.  Hence if $\alpha^*(\phi)>1$, the constrained solution is $\alpha(\phi)=1$, giving the global solution \eqref{eq:tuning} from the theorem statement.

Substituting \eqref{eq:tuning} into \eqref{eq:W2} and \eqref{eq:ET} shows that the resulting values of $\E(W^2)$ and $\E(T)/N$ are continuous in $\lambda$.  Furthermore all $\E(T)/N$ values in $\Upsilon$ are attainable by \eqref{eq:tuning} under some choice of $\lambda$.  Hence given $\upsilon \in \Upsilon$ there is some $\mu^*$ for which the solution to $P(\upsilon, \mu^*)$ has $\E(T)/N=\upsilon$.  This must also be a solution to $P(\upsilon)$ since otherwise a superior choice of $\alpha$ for $P(\upsilon)$ is also superior for $P(\upsilon, \mu^*)$.  Now choose $\upsilon^*$ so that the solution to $P(\upsilon^*)$ minimises $\E(W^2) \E(T)/N$.  This must be a solution to $P$ since otherwise a superior choice of $\alpha$ for $P$ is superior to the solution already found for some $P(\nu)$.

\section{Tuning g} \label{sec:g}

This section proves Corollary \ref{cor:ABC-IS}.
First the RW-IS case is considered.
Recall that in this case $\gamma(\theta) = \E(\hat{L}^2 | \theta)$.

RW-IS can be seen as a special case of LIS where $\phi=\theta$, $\bar{T}_1(\theta)=0$ and $\bar{T}_2(\phi) = \bar{T}(\theta)$.
Repeating the working of Appendix \ref{sec:alphaproof} to optimise the choice of $\alpha(\theta) g(\theta)$ gives the unconstrained solution:
\begin{equation} \label{eq:opt alpha g}
\alpha(\theta) g(\theta) = \lambda \pi(\theta) \left[\frac{\gamma(\theta)}{\bar{T}(\theta)}\right]^{1/2}.
\end{equation}
Various choices of $\alpha$, such as $\alpha \equiv 1$, give a solution which also meets the constraint on $\alpha$.  These all give algorithms which are equivalent to RW-IS with $g(\theta) \propto \pi(\theta) \left[\dfrac{\gamma(\theta)}{\bar{T}(\theta)}\right]^{1/2}$ as claimed.

Finally, to prove the ABC-IS case, note that this is a special case of RW-IS with $\hat{L} = \hat{L}_{\text{ABC}}$ so the same result holds.
Since $\hat{L}_{\text{ABC}}$ is Bernoulli, $\gamma(\theta) = \E(\hat{L}^2 | \theta) = \E(\hat{L}_{\text{ABC}} | \theta)$ as claimed.

\bibliography{lazy}
\end{document}